\documentclass[journal]{IEEEtran}
\usepackage{amsmath,amsfonts}
\usepackage[thmmarks,amsmath]{ntheorem}
\usepackage{algorithmic}
\usepackage{algorithm}
\usepackage{array}
\usepackage[caption=false,font=normalsize,labelfont=sf,textfont=sf]{subfig}
\usepackage{textcomp}
\usepackage{stfloats}
\usepackage{url}
\usepackage{verbatim}
\usepackage{graphicx}
\usepackage{cite}
\usepackage{amssymb}
\usepackage{physics}
\usepackage{epstopdf}
\usepackage{cases}
\usepackage{indentfirst}
\usepackage{amsmath}
\allowdisplaybreaks[4]

\newtheorem*{proof}{Proof:}
\theoremseparator{:}

\newtheorem{theorem}{Theorem}
\newtheorem{lemma}{Lemma}

\begin{document}

	\title{Sensing-Communication-Computing-Control Closed-Loop Optimization for 6G Unmanned Robotic Systems}

	\author{Xinran~Fang, Chengleyang~Lei,
		Wei~Feng,~\IEEEmembership{Senior Member,~IEEE,} Yunfei~Chen,~\IEEEmembership{Senior Member,~IEEE,}\\ Ming Xiao,~{\IEEEmembership{Senior Member,~IEEE,}} Ning~Ge, and Cheng-Xiang Wang, \emph{Fellow, IEEE}
		\thanks{X.~Fang, C.~Lei, W.~Feng, and N. Ge  are with the State Key Laboratory of Space Network and Communications, Department of Electronic Engineering, Tsinghua University, Beijing 100084, China (e-mail: {fxr20}@mails.tsinghua.edu.cn, {lcly21}@mails.tsinghua.edu.cn, {fengwei}@tsinghua.edu.cn, and  {gening}@tsinghua.edu.cn).

		Y. Chen is with the Department of Engineering, University of Durham,
		DH1 3LE Durham, U.K. (e-mail: yunfei.chen@durham.ac.uk).

		M. Xiao is with the Department of Information Science and Engineering, School of Electrical Engineering and Computer Science, Royal Institute of Technology, Sweden (e-mail: mingx@kth.se).

		C.-X. Wang is with the National Mobile Communications Research Laboratory, School of Information Science and Engineering, Southeast University, Nanjing, 210096, China, and also with the Purple Mountain Laboratories, Nanjing 211111, China (e-mail: chxwang@seu.edu.cn).}}

	\maketitle

	\begin{abstract}
	Rapid advancements in field robots have brought a new kind of cyber physical system (CPS)--unmanned robotic system--under the spotlight. In the upcoming sixth-generation (6G) era, these systems hold great potential to replace humans in hazardous tasks. This paper investigates an unmanned robotic system comprising a multi-functional unmanned aerial vehicle (UAV), sensors, and actuators. The UAV carries communication and computing modules, acting as an edge information hub (EIH) that transfers and processes information. During the task execution, the EIH gathers sensing data, calculates control commands, and transmits commands to actuators—leading to reflex-arc-like sensing-communication-computing-control ($\mathbf{SC}^3$) loops. Unlike existing studies that design $\mathbf{SC}^3$ loop components separately, we take each $\mathbf{SC}^3$ loop as an integrated structure and propose a goal-oriented closed-loop optimization scheme. This scheme jointly optimizes uplink and downlink (UL\&DL) communication and computing within and across the $\mathbf{SC}^3$ loops to minimize the total linear quadratic regulator (LQR) cost. We derive optimal closed-form solutions for intra-loop allocation and propose an efficient iterative algorithm for inter-loop optimization. Under the condition of adequate CPU frequency availability, we derive an approximate closed-form solution for inter-loop bandwidth allocation. Simulation results demonstrate that the proposed scheme achieves a two-tier task-level balance within and across $\mathbf{SC}^3$ loops.
	\end{abstract}

	\begin{IEEEkeywords}
		Closed-loop optimization, goal-oriented communication, sensing-communication-computing-control ($\mathbf{SC}^3$) loop, unmanned robotic system, uplink and downlink (UL\&DL) configuration
	\end{IEEEkeywords}

	\section{Introduction}
	\subsection{Background and Motivation}
	Cyber physical system (CPS) refers to the system that integrates communication, computing, and control to get information from the physical world, perform analysis, and change the physical world \cite{ref1.1,ref1.2}. In recent decades,
	a new kind of CPSs--unmanned robotic system--has received great attention for their human-like capabilities \cite{Jain}. Driven by technologies such as smart computing and cloud-fog architectures, these systems have made significant strides in industrial automation, significantly improving production efficiency  \cite{Lyu3,Jin}. Moreover, recent advancements in field robots have extended unmanned robotic systems from indoor factories to outdoor environments. These systems are now being deployed in hard-to-access areas, replacing humans in dangerous tasks such as disaster rescue \cite{ref3}, oil exploitation, and space exploration \cite{Fang}.
	To fully unlock the potential of unmanned robotic systems, supporting field robots has been identified as an important use case for the sixth-generation (6G) network \cite{nextG}.

	Due to the lack of ground facilities in remote areas, aerial-borne and space-borne platforms such as unmanned aerial vehicles (UAVs) and satellites, are envisioned to provide global seamless coverage in 6G \cite{feng, feng2}. These platforms are agile to carry communication and computing modules, acting as edge information hubs (EIHs) to both transfer and process information \cite{Lei2}. A digital twin can be further integrated into EIHs to enable intelligent decision-making \cite{Jones}.
	During the task execution, sensors collect raw data and upload them to the EIH via the sensor-EIH link. Based on the sensing data, the EIH calculates control commands.
	These commands are then transmitted to the actuators via the EIH-actuator link for actions.	Synergistically, the sensor, the sensor-EIH link, the EIH, the EIH-actuator link, and the actuator form an integrated sensing-communication-computing-control ($\mathbf{SC}^3$) loop. Through effective feedback, the $\mathbf{SC}^3$ loop continuously learns the behavior of the physical system and guides its evolution to the desired direction. From the perspective of biology, the $\mathbf{SC}^3$ loop has great similarity to the reflex arc, as we compared in Fig. \ref{fig1}.
	In biology, it is widely recognized that the presence of the reflex depends on the functional integrity of the reflex arc \cite{bio}. In this sense, it is reasonable to take the $\mathbf{SC}^3$ loop as an integrated structure when we investigate the unmanned robotic system.
	\begin{figure*} [t]
		\centering
		\includegraphics[width=0.8\linewidth]{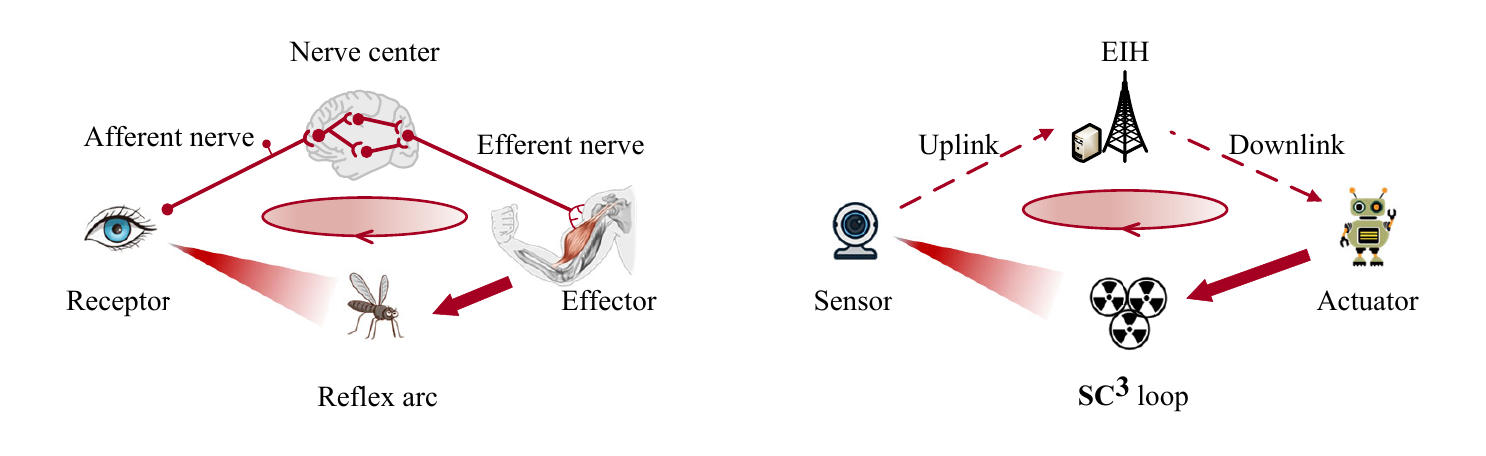}
		\caption{Comparisons between the reflex arc and the $\mathbf{SC}^3$ loop. The reflex arc consists of five parts: receptor, afferent nerve, nerve center, efferent nerve, and effector. By analogy, the $\mathbf{SC}^3$ loop also consists of five parts: sensor, uplink, EIH, downlink, and actuator. The similarity of these two structures motivates us to take the $\mathbf{SC}^3$ loop as an integrated structure and devise the unmanned robotic system from a structured lens.}
		\label{fig1}
	\end{figure*}

	However, due to the heterogeneity of communication, computing, and control, current applications design $\mathbf{SC}^3$-loop components separately. For this reason, the fifth-generation (5G) network positions itself as a communication network and primarily focuses on data transmission. It applies multiple access and duplex  techniques, e.g., orthogonal frequency-division multiple access (OFDMA) and time-division duplexing/frequency-division duplexing (TDD/FDD), to decompose the network into independent links. Although this link-level division brings high capacity for communication, the task-level connections of different components within the $\mathbf{SC}^3$ loop are  disintegrated. The mismatch of the sensor-EIH link, computing, and EIH-actuator link impairs the overall functioning of the $\mathbf{SC}^3$ loop, ultimately degrading the task efficiency of the robotic system.

	Unlike Shannon who separated communication from its served system and regarded communication engineering as ``\emph{reproducing at one point either exactly or approximately a message selected at another point}"  \cite{Shannon}, the creator of cybernetics, Wiener, noted that \emph{``The problems of control engineering and communication engineering were inseparable, and that they centered not around the technique of electrical engineering but around the much more fundamental notion of the message.}" \cite{wiener}.  By understanding communication through Wiener’s lens, we can find that communication is not merely for data transmission but an integral part for task execution.  In this context, we need to shift our focus from  individual communication links to integrated $\mathbf{SC}^3$ loops. As a result, the communication efficiency transcends bit transmission: it involves cooperating with computing and control to achieve good task performance. Building on these insights, we use a structured lens to regard the sensor and actuator within the $\mathbf{SC}^3$ loop as a virtual user \cite{feng1}. The sensor-EIH link and EIH-actuator link are thus the uplink and downlink  (UL\&DL) of this virtual user. On this basis, we consider the goal-oriented closed-loop design for the unmanned robotic system.


	\subsection{Related Studies}

	In the literature, related studies mostly focus on part of the $\mathbf{SC}^3$ loop. In wireless control systems (WCSs), the interplay between communication and control has been extensively studied. In mobile edge computing (MEC), joint optimization of communication and computing has been a central topic. In goal-oriented communication, related studies integrated the task efficiency into the communication design, and the learning-based approaches have become a dominant trend. In the following, we review related studies on WCS, MEC, and goal-oriented communication.

	{\bf{\emph{1) WCS}}}: 
	In WCS, related studies investigated the impact of imperfect communication, such as low data rate and latency, on  control \cite{Park1,Wang}. From the theoretical perspective, Tatikonda \emph{et al.}  analyzed linear discrete-time systems and derived the lower and upper bounds on the data rate required for different control objectives \cite{Tatikonda}. 
	Kostina \emph{et al.} further generalized the work in \cite{Tatikonda} and derived the lower bound of the data rate given the control objective measured by the linear quadratic regulator (LQR) cost \cite{Kostina}. In terms of optimization schemes, Baumann \emph{et al.} and Gatsis \emph{et al.} devised control-aware transmission schemes that the sensor uploads sensing data only when the estimation error surpasses a certain threshold \cite{Baumann,Gatsis}. 
	By taking the LQR cost as the objective, Wang \emph{et al.} optimized the UL \cite{Wang1}, while Lei \emph{et al.} and Fang \emph{et al.} optimized the DL \cite{Lei,Fang}. 
	For ultra-reliable and ultra-low latency communication (URLLC)-supported WCSs, Chang \emph{et al.} jointly optimized communication bandwidth, power, and control convergence rate to maximize the spectrum efficiency (SE) \cite{Chang}, and Yang \emph{et al.}  optimized the UL\&DL transmit power and block length to minimize a new metric named energy-to-control efficiency \cite{Yang}. In addition to above theoretical advancements, researchers also established simulation platforms to testify the control performance under different communication conditions. Bhimavarapu \emph{et al.} proposed a Unobtrusive Latency Tester solution to measure the communication latency and reliability in real-time control \cite{Bhimavarapu}. Lyu \emph{et al.} proposed a novel hardware-in-the-loop simulation platform, which provided reliable test results by accurately stimulating the real wireless environment \cite{Lyu,Lyu2}. Utilizing these platforms, the performance of 5G and Wi-Fi 6 was assessed, which provides valuable insights for selecting communication protocols for specific applications \cite{Bhimavarapu,Lyu,Lyu2}.

	{\bf{2) \emph{MEC}}}: The core idea of MEC is to place computing resources close to mobile devices to support real-time applications. Most studies used joint communication and computing optimization to minimize the task latency \cite{Liu,Ferdouse,Chen}, energy consumption \cite{Cao}\cite{Wen2}, or maximize the utility \cite{Jian, Li,Wang2}. For example, 
	Wen \emph{et al.} considered a TDD-mode MEC system in which each user is allocated a single slot to complete the data uploading, computing, and result downloading. The authors optimized bit, subchannel, and time allocation to minimize energy consumption \cite{Wen2}.  
	Jian \emph{et al.} focused on computing and DL transmission, proposing a joint MEC server-user association, CPU frequency, and bandwidth optimization scheme \cite{Jian}.
%
	 Wang \emph{et al.} considered a joint sensing, communication, and computing framework, where the base station detects surrounding objects and uploads part of computing tasks to the cloud center using integrated sensing and communication (ISAC) signals. The authors optimized the beamformer and CPU frequency to maximize the computing rate \cite{Wang2}.

{\bf\emph{3) Goal-Oriented Communication}}: The core idea of goal-oriented communication is to adapt communication strategies to the specific needs of supported tasks.
For example, Wen \emph{et al.} considered a real-time inference task performed by ISAC devices and an edge server. The authors took the discriminant gain as the objective and jointly optimized the transmit power, time allocation, and quantization bit allocation \cite{Wen}.
Girgis \emph{et al.} proposed a semantic communication-control co-design that trains an encoder at the edge to abstract the low dimensional feature and a decoder at the control center to recover the state and calculate the command  \cite{Girgis}. 
Shao \emph{et al.} utilized the information bottleneck method to maximize the mutual information between the inference result and the encoded feature, while minimizing the mutual information between the raw input data and the encoded feature \cite{Shao}. 
Mostaani \emph{et al.} devised a distributed deep learning scheme for a multi-agent system, and maximized the long-term return by jointly optimizing  communication and control policies  \cite{Mostaani}.


In summary, these studies offer valuable insights into the design of the $\mathbf{SC}^3$ loop. Most of these studies have focused on part of the $\mathbf{SC}^3$ loop, such as UL and computing or DL and control, using indirect task metrics like latency or computing performance as their optimization objectives.
However, for unmanned robotic systems, the reflex-arc-like $\mathbf{SC}^3$ loop is an integrated structure. The loop performance is reflected by the control actions on the physical world rather than the intermediate metrics. There is a lack of work that jointly optimizes the UL, computing, and DL within and across the $\mathbf{SC}^3$ loops from a goal-oriented perceptive. 

\subsection{Main Contributions}
In this paper, we consider an unmanned robotic system formed by a multi-functional UAV, sensors, and actuators. The UAV carries communication and computing modules, with a digital twin simulating the physical process in real time. During the task execution, the UAV acts as an EIH that collects data, calculates commands, and distributes commands to actuators to take actions. Together, the EIH, sensors, and actuators form reflex-arc-like $\mathbf{SC}^3$ loops. Different from current studies that devise $\mathbf{SC}^3$ loop components separately, our study takes the $\mathbf{SC}^3$ loop as an integrated structure and jointly configures UL, computing, and DL within and across the $\mathbf{SC}^3$ loops from a goal-oriented perspective. The main contributions are listed as follows.
\begin{enumerate}
	\item  We investigate an unnamed robotic system formed by multiple reflex-arc-like $\mathbf{SC}^3$ loops. To effectively use limited resources, we propose a goal-oriented closed-loop optimization scheme that jointly optimizes bandwidth, time, and CPU frequency within and across  $\mathbf{SC}^3$ loops, with the objective of minimizing the total LQR cost.
	\item In the intra-loop configuration, we derive optimal closed-form solutions for UL\&DL bandwidth and time, along with the optimal LQR cost. Based on these results, we demonstrate the task-level balance between the UL\&DL, as well as the interchange relationships between communication bandwidth and computing CPU frequency.
	\item In the inter-loop configuration,  we propose an iterative algorithm to optimize the bandwidth and CPU-frequency allocation. Under the condition of adequate CPU frequency availability, we derive an approximate closed-form solution for the inter-loop bandwidth allocation and analyze the allocation principles regarding communication and control parameters.
	\item We conduct comprehensive simulation to validate our findings. We show the superiority of the goal-oriented closed-loop optimization by comparing it with separate schemes and communication-oriented schemes. We show that the proposed scheme achieves a two-tier task-level balance within and across the $\mathbf{SC}^3$ loops, which is crucial for the overall performance of the unmanned robotic system.
\end{enumerate}

\subsection{Organization and Notation}
The rest of this paper is organized as follows. Section \ref{section 2} introduces the model of the unmanned robotic system and the related $\mathbf{SC^3}$ loop. Section \ref{section 3} presents the goal-oriented closed-loop optimization scheme and its solution. Section \ref{section 4} presents simulation results and discussion.  Section \ref{section 5} draws conclusions.

Throughout this paper, vectors, matrices, and sets are represented by bold lowercase letters, bold uppercase letters, and curly uppercase letters, respectively. $\mathbb{R}^{n \times n}$ represents the set of $n\times n$ real matrices, $\mathbf{I}_n$ is the $n\times n$ unit matrix, and $\mathbf{0}_n$ is the $n\times n$ zero matrix. $\lambda(\mathbf{A})$ denotes the eigenvalue of matrix $\mathbf{A}$ and $\det \mathbf{A}$ denotes the determinant of matrix $\mathbf{A}$. The complex Gaussian distribution of zero mean and $\sigma^2$ variance is denoted as $\mathcal{CN}(0,\sigma^2)$. The optimal solution of $x$ is denoted as $(x)^*$.

\section{Unmanned Robotic System and $\mathbf{SC}^3$ Loop Model}
\label{section 2}
\begin{figure} [t]
	\centering
	\includegraphics[width=1\linewidth]{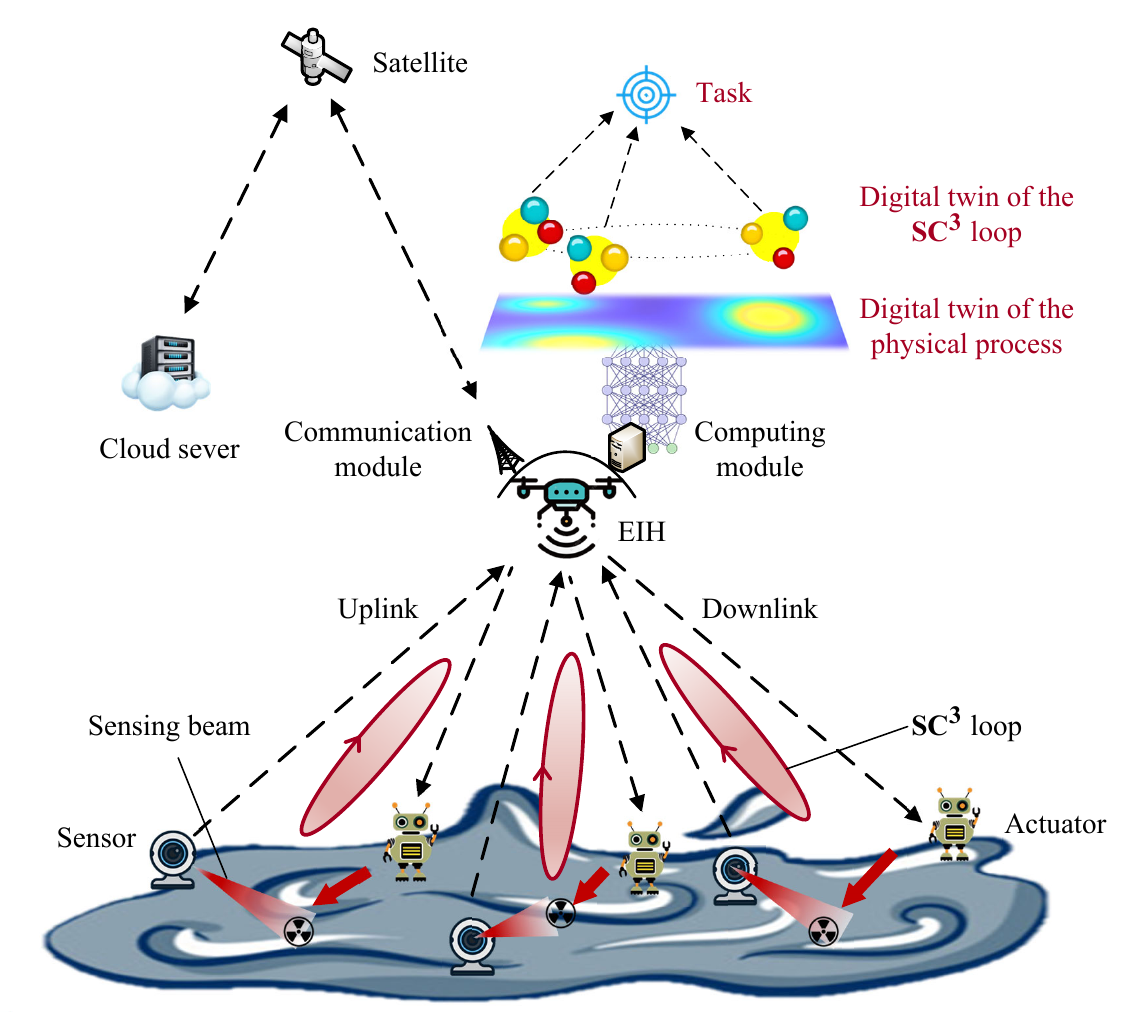}
	\caption{Illustration of the unmanned robotic system. The system comprises a multi-functional UAV and $K$ pairs of sensors and actuators, which synergistically form $K$ $\mathbf{SC}^3$ loops. The UAV carries communication and computing modules, acting as an EIH that transfers and processes information. A digital twin is integrated within it to simulate the physical process in real time.  In this figure, the objects marked in black font represent the physical world and the objects marked in red font represent the digital twin.   }
	\label{model}
\end{figure}

	As shown in Fig. \ref{model}, we consider an unmanned robotic system performing a control-type task, such as radioactive object recycling.  The system consists of $K$ $\mathbf{SC}^3$ loops, which are formed by a multi-functional UAV and $K$ pairs of sensors and actuators. The UAV carries communication and computing modules, serving as an EIH that transfers and processes information. A digital twin can be further integrated into the EIH for simulating the physical process in real time. We regard each pair of sensor and actuator as a virtual user. In this way, the sensor-EIH link and EIH-actuator link are the UL\&DL of the virtual user. During the task execution, sensors collect raw data and upload them to the EIH via the UL. Based on the sensing data, the EIH updates the digital twin and calculates control commands using the reasoning of system behavior.  These commands are transmitted to the actuators via the DL. Then, the actuators take actions.  To ensure the correct operations of the EIH, the satellite provides telemetry, tracking, and control services for the EIH, while a cloud server provides high-level guidelines. In this paper, we focus on the unmanned robotic system in the field and do not consider the remote satellite and cloud server.

	For simplicity, we assume that the controlled systems are linear time-invariant processes. For nonlinear processes, this linear model represents the linearization around the system's working point, which provides a reasonable approximation of its local behavior \cite{Lal}.  Taking the $k$-th $\mathbf{SC}^3$ loop and time index $i$ as an example, the system evolution is described by the following equation:
	\begin{equation}
		\label{sta_elv}
		\mathbf{x}_{k,i+1} = \mathbf{A}_k\mathbf{x}_{k,i}+\mathbf{B}_k\mathbf{u}_{k,i}+\mathbf{v}_{k,i},
	\end{equation}
	where $\mathbf{x}_{k,i}\in\mathbb{R}^{n\times1}$ denotes the system state, $\mathbf{u}_{k,i}\in\mathbb{R}^{n\times1}$ denotes the control action, $\mathbf{v}_{k,i}\in\mathbb{R}^{n\times1}$ denotes the process noise, and $n$ denotes the dimension of the controlled system. The matrices $\mathbf{A}_k\in\mathbb{R}^{n\times n}$ and $\mathbf{B}_k\in\mathbb{R}^{n\times n}$ are determined by the system dynamics, i.e., $\mathbf{A}_k$ quantifies the intrinsic dynamics of the system without external intervention, and $\mathbf{B}_k$ quantifies the effect of control actions on the state evolution. The system to be controlled  is inherently unstable $(\lambda(\mathbf{A}_k)>1)$, and  it can be stabilized by the pair $(\mathbf{A}_k,\mathbf{B}_k)$. We use LQR cost to measure the control performance of the $\mathbf{SC}^3$ loop, which is a weighted summation of system states and control inputs over the entire control process,
	\begin{equation}
		l_k=	\limsup\limits_{N\rightarrow \infty}\mathbb{E} \left[ \sum_{i=1}^{N} \left(\mathbf{x}_{k,i}^\text{T}\mathbf{Q}_k\mathbf{x}_{k,i} +\mathbf{u}_{k,i}^\text{T}\mathbf{R}_k\mathbf{u}_{k,i}\right) \right],
	\end{equation}
	where $l_k$ denotes the LQR cost, $\mathbf{Q}_k\in\mathbb{R}^{n\times n}$ and $\mathbf{R}_k\in\mathbb{R}^{n\times n}$ are weighting matrices that balance the cost between state deviations and control inputs. 

	For communication, we denote the bandwidth, transmission time, and the SE of UL\&DL as $\mathcal{B}^u=\{B_k^u\}_{k=1}^K$, $\mathcal{B}^d=\{B_k^d\}_{k=1}^K$, $\mathcal{T}^u=\{t^u_k\}_{k=1}^K$, $\mathcal{T}^d=\{t^d_k\}_{k=1}^K$, $\mathcal{R}^u=\{r^u_k\}_{k=1}^K$, and $\mathcal{R}^d=\{r^d_k\}_{k=1}^K$, respectively.
	According to the Shannon capacity, $r^u_k $ and $r^d_k$ are calculated as follows
	\begin{equation}
		\label{speed}
		\begin{aligned}
			r^u_k&=\log_2(1+\frac{|h^u_k|^2p^u_k}{\sigma^2})  \ \text{(bits/s/Hz)},\\
			r^d_k&=\log_2(1+\frac{|h^d_k|^2p^d_k}{\sigma^2})  \ \text{(bits/s/Hz)},
		\end{aligned}
	\end{equation}
	where $h^u_k$ and $h^d_k$ denote the channel gain, $p^u_k$ and $p^d_k$ denote the transmit power, $\sigma^2$ is the channel noise variance. 
	In this paper, we assume that the UL\&DL SEs, $r^u_k$ and $r^d_k$, remain constant throughout the control process. This can be achieved through adaptive power control, that the transmit power is dynamically adjusted according to channel conditions to maintain a target signal-to-noise ratio (SNR) \cite{Mallik}. This assumption applies to the bandwidth-constrained systems whose power resources are relatively sufficient to compensate for channel variations. By making this assumption, we focus on the bandwidth and time allocation in this paper.
	The channel gain consists of both small-scale fading and large-scale fading, which is given by,
	\begin{equation}
		\begin{aligned}
		h^u_k&=\beta^u_k\sqrt{g^u_k},\\
		h^d_k&=\beta^d_k \sqrt{g^d_k},\\
		\end{aligned}
	\end{equation}
	where $\beta^u_k$ and $\beta^d_k$ are the small-scale fading, which conforms to the complex Gaussian distribution, and $g^u_k$ and $g^d_k$ are the large-scale fading, which are calculated by the path-loss model,
	\begin{equation}
		\label{channel_gain}
		\begin{aligned}
		g^u_k \ (\text{dB})&=32.4+20\log_2(d^u_k)+20\log_2(f_c), \\
	    g^d_k\ (\text{dB})&=32.4+20\log_2(d^d_k)+20\log_2(f_c), \\
		\end{aligned}
	\end{equation}
	where $d^u_k$ (km) and $d^d_k$ (km) denote the transmission distance of the UL\&DL, and $f_c$ (MHz) denotes the carrier frequency. In each $\mathbf{SC}^3$ cycle, the amount of information transmit via the UL\&DL are given by,
	\begin{equation}
		\label{bit}
		\begin{aligned}
			D^u_k&=B^u_kt^u_kr^u_k \ \text{(bits)},\\
			D^d_k&=B^d_kt^d_kr^d_k \ \text{(bits)}.
		\end{aligned}
	\end{equation}

	For computing, upon receiving sensing data, the EIH processes these data, extracts task-related information, and updates the digital twin. Using the updated digital twin, which simulates system behavior, the EIH calculates control commands.   The information extraction process is described as,
	\begin{equation}
		D^u_k\rightarrow \rho_k D^u_k \ \text{(bits)},
	\end{equation}
	where $\rho_k$ is a proportion parameter. $\rho_k D^u_k$ denotes the  task-related information extracted from the sensing data, which also denotes the information contained in the command. The computing time is calculated by,
	\begin{equation}
		t^{\text{comp}}_k=\frac{\alpha_kD^u_k}{f_k},
	\end{equation}
	where $\alpha_k$ (cycles/bit) denotes the processing difficulty (the required CPU cycles for processing one-bit data), $f_k$ denotes the allocated  CPU frequency for the $k$-th $\mathbf{SC}^3$ loop, and we denote $\mathcal{F}=\{f_k\}_{k=1}^K$.
	In practice, $\alpha_k$ and $\rho_k$ are the parameters of the neural network . A larger network usually provides a more accurate modeling of the physical process, enables a more thorough extraction of task-related information, and outputs more effective control commands, resulting in a higher $\rho_k$. Conversely, the larger network also incurs greater computational complexity, resulting in a higher $\alpha_k$. Therefore, finding a proper precision of the digital twin is crucial for the practical deployment \cite{liu}. Nonetheless, we do not delve into this issue in this paper. The calculated commands are sent to the actuators to guide their actions.
	However, due to the capacity constraints of the DL, the command information may not fully reach the actuators. The information that finally works for the controlled system is jointly determined by the task-related information extracted from the UL and the information successfully transmitted via the DL,
	\begin{equation}
		\label{SC3}
		D^{\mathbf{SC}^3}_k\leqslant \min\{\rho_k D^u_k,D^d_k\} \ \text{(bits)},
	\end{equation}
	where $D^{\mathbf{SC}^3}_k$ is defined as the closed-loop information, which is the information that finally works within one $\mathbf{SC}^3$ cycle. According to \cite{Kostina}, the lower bound of the LQR cost has a direct relationship with the closed-loop information as,
	\begin{equation}\label{f6}
		l_k \geqslant \frac{n N \!\left( \mathbf{v_k}\right)|\det \mathbf{M}_k|^\frac{1}{n}} {2^{\frac{2}{n}(D^{\mathbf{SC}^3}_k-\log_2|\det \mathbf{A}_k|)}-1}+\text{tr}\left( \mathbf{\Sigma}_{\mathbf{v}_k}\mathbf{S}_k\right),
	\end{equation}
	where $N(\mathbf{x})\triangleq\frac{1}{2\pi e}e^{\frac{2}{n}h(\mathbf{x})}$ and $h(\mathbf{x})$ is the differential entropy of $\mathbf{x}$, i.e., $h(\mathbf{x})\triangleq-\int_{\mathbb{R}^n}f_{\mathbf{x}}(x)\log  f_{\mathbf{x}}(x)\mathrm{d}x$, $\mathbf{\Sigma}_{\mathbf{v}_k}$ is the covariance matrix of the process noise, $\log_2|\det\mathbf{A}_k|$ is the intrinsic entropy, and $\mathbf{M}_{k}$ and $\mathbf{S}_k$ are determined by the Riccati equations,
	\begin{equation}\label{Riccati}
		\begin{aligned}
			\mathbf{S}_k & = \mathbf{Q}_k + \mathbf{A}_k^T\left(\mathbf{S}_k- \mathbf{M}_k\right) \mathbf{A}_k,\\
			\mathbf{M}_k & = \mathbf{S}_k^T \mathbf{B}_k \left( \mathbf{R}_k + \mathbf{B}_k \mathbf{S}_k \mathbf{B}_k\right)^{-1} \mathbf{B}_k^\text{T} \mathbf{S}_k.
		\end{aligned}
	\end{equation}
	To ensure the system can be stabilized, the closed-loop information needs to satisfy the following stable condition \cite{Kostina},
	\begin{equation}
		\label{stable}
		D^{\mathbf{SC}^3}_k> \log_2|\det\mathbf{A}_k|.
	\end{equation}
	In addition, the $\mathbf{SC}^3$ loop needs to run within the given cycle time, and we have the following cycle-time constraint,
	\begin{equation}
		\label{cycle-time}
		t^u_k+t^{\text{comp}}_k+t^d_k\leqslant T_k,
	\end{equation}
	where $T_k$ denotes the cycle time for the $k$-th $\mathbf{SC}^3$ loop.
	We can see from \eqref{SC3} and \eqref{f6}, the LQR cost is constrained by the closed-loop information, and the closed-loop information is further determined by the information transmitted via the UL\&DL. In addition, the computing time affects the time available for data transmission, thereby indirectly influencing the LQR cost. Consequently, a co-design of UL, computing, and DL is essential to ensure a good control performance of the $\mathbf{SC}^3$ loop.

	\section{Goal-Oriented Closed-Loop Optimization}
	\label{section 3}
	\subsection{Problem Formulation}
	In this paper, we aim to minimize the total LQR cost of the unmanned robotic system by jointly optimizing bandwidth, CPU frequency, and time within and across the $\mathbf{SC}^3$ loops. The optimization problem is formulated as follows:
		\begin{subequations}
		\begin{align}
			\mbox{(P1)} \ \ &\min\limits_{\mathcal{B}^u,\mathcal{T}^u,\mathcal{F},\mathcal{B}^d,\mathcal{T}^d} \sum_{k=1}^Kl_k \\
			\text{s.t.} \ &l_k \geqslant \frac{n N \!\left( \mathbf{v}_k\right)|\det \mathbf{M}_k|^\frac{1}{n}} {2^{\frac{2}{n}(D^{\mathbf{SC}^3}_k-\log_2|\det \mathbf{A}_k|)}-1}+\text{tr}\left( \mathbf{\Sigma}_{\mathbf{v}_k}\mathbf{S}_k\right), \forall k \label{a1} \\
			&D^{\mathbf{SC}^3}_k\leqslant \min(\rho_k D^u_k,D^d_k), \forall k \label{SC31}\\
			&D^{\mathbf{SC}^3}_k> \log_2|\det \mathbf{A}_k|, \forall k \label{14d}\\
			&D^u_k\leqslant t^u_kB^u_kr^u_k, \forall k \label{ulc} \ \\
			&D^d_k\leqslant t^d_kB^d_kr^d_k, \forall k \label{dlc}\\
			&t^u_k+\frac{\alpha_k D^u_k}{f_k}+t^d_k\leqslant T_k, \forall k \label{clc}\\
			&\sum\limits_{k=1}^K(B^u_k+B^d_k)\leqslant B_{\max}  \\
			&\sum\limits_{k=1}^Kf_k\leqslant f_{\max} \label{f13}\\
			&B^u_k\geqslant 0\ t^u_k\geqslant 0\ f_k\geqslant 0\ B^d_k\geqslant 0\ t^d_k\geqslant 0, \ \forall k,
		\end{align}
	\end{subequations}
	where $B_{\max}$ and $f_{\max}$ denote the maximal bandwidth and CPU frequency, respectively.
	Due to the coupling relationships among different variables, (P1) is a highly complex problem with non-convex constraints \eqref{ulc}, \eqref{dlc} and \eqref{clc}. To address this challenge, we begin by investigating the intra-loop configuration within an $\mathbf{SC}^3$ loop and move to investigate the inter-loop optimization across  multi-$\mathbf{SC}^3$ loops.
	\subsection{Single-$\mathbf{SC}^3$-Loop Optimization}
	For simplicity, the index $k$ is omitted when we consider an $\mathbf{SC}^3$ loop. When $K=1$, the optimization problem becomes
	\begin{subequations}
		\begin{align}
			\mbox{(PA-1)} \ \ &\min\limits_{B^u,t^u,f,B^d,t^d} l \\
			\text{s.t.} \ &l \geqslant \frac{n N \!\left( \mathbf{v}\right)|\det \mathbf{M}|^\frac{1}{n}} {2^{\frac{2}{n}(D^{\mathbf{SC}^3}-\log_2|\det \mathbf{A}|)}-1}+\text{tr}\left( \mathbf{\Sigma}_{\mathbf{v}}\mathbf{S}\right) \label{15b} \\
			&D^{\mathbf{SC}^3}\leqslant \min(\rho D^u,D^d) \label{15c} \\
			&D^{\mathbf{SC}^3}> \log_2|\det \mathbf{A}| \label{15d}\\
			&D^u\leqslant t^uB^ur^u \label{15e} \\
			&D^d\leqslant t^dB^dr^d \label{15f} \\
			&t^u+\frac{\alpha D^u}{f}+t^d\leqslant T \\
			&B^u+B^d\leqslant B_{\max} \\
			&f\leqslant f_{\max}, \label{4-f13}\\
			&B^u\geqslant 0\ t^u\geqslant 0\ f\geqslant 0\ B^d\geqslant 0 \ t^d\geqslant 0.
		\end{align}
	\end{subequations}
	As the lower bound of the LQR cost \eqref{15b} monotonically decreases with the closed-loop information, minimizing the LQR cost is equivalent to maximizing the closed-loop information for an $\mathbf{SC}^3$ loop. In addition, as there are no other $\mathbf{SC}^3$ loops competing for resources, the optimal CPU frequency is its maximum allowable value, i.e., $(f)^*=f_{\max}$.  As a result, the critical focus in the single-$\mathbf{SC}^3$-loop optimization is the UL\&DL configuration,
		\begin{subequations}
		\begin{align}
			\mbox{(PA-2)} \ \ &\max\limits_{B^u,t^u,B^d,t^d} D^{\mathbf{SC}^3} \\
			\text{s.t.} \ \
			&D^{\mathbf{SC}^3}\leqslant \min(\rho D^u,D^d) \label{16b} \\
			&D^u\leqslant t^uB^ur^u \label{16c}\\
			&D^d\leqslant t^dB^dr^d \label{16f}\\
			&t^u+\frac{\alpha D^u}{f_{\max}}+t^d\leqslant T \label{16e} \\
			&B^u+B^d\leqslant B_{\max}\\
			&B^u\geqslant 0\ t^u\geqslant 0\ B^d\geqslant 0 \ t^d\geqslant 0,
		\end{align}
	\end{subequations}
	where the stable condition \eqref{15d} is omitted. We test it after solving (PA-2). If the optimization result, $(D^{\mathbf{SC}^3})^*$ satisfies \eqref{15d}, we calculate the LQR cost according to \eqref{f6}. Otherwise, the controlled system cannot be stabilized, and the LQR cost is infinite.
	\begin{lemma}
		The optimal UL\&DL configuration for an $\mathbf{SC}^3$ loop is to keep a task-level balance, described by the following equation,
		\begin{equation}
			\label{equal}
			\rho (D^u)^*=(D^d)^*.
		\end{equation}
		The optimal UL\&DL time allocation is given by,
		\begin{equation}
			\label{opt}
			\begin{aligned}
				(t^u)^*=\frac{\frac{1}{\rho B^ur^u}}{\frac{1}{\rho B^ur^u}+\frac{\alpha}{\rho f_{\max}}+\frac{1}{B^dr^d}}T \\ (t^d)^*=\frac{\frac{1}{B^dr^d}}{\frac{1}{\rho B^ur^u}+\frac{\alpha}{\rho f_{\max}}+\frac{1}{B^dr^d}}T.
			\end{aligned}
		\end{equation}
	\end{lemma}
	\begin{proof}
		See Appendix A.
	\end{proof}
	From \eqref{equal}, we can see that the task-level balance between UL\&DL is to match the task-related information extracted from the raw data transmitted via the UL with the information transmitted via the DL. This balance can be illustrated by the analogy to water pipes. The UL and computing are considered as one pipe, while the DL is another. Just as the maximum water flow through the interconnected pipes is limited by the pipe with the smallest capacity, the optimal solution requires equalizing the capacities of both pipes, as indicated by the equilibrium in \eqref{equal}.

	Based on the \eqref{equal} and \eqref{opt}, we can express the closed-form information as an expression of the bandwidth,
	\begin{equation}
		\begin{aligned}
		D^{\mathbf{SC}^3}&=\min\{\rho (D^u)^*,(D^d)^*\}=\rho (D^u)^*\\
		&=\rho(t^u)^*B^ur^u=\frac{T}{\frac{1}{\rho B^ur^u}+\frac{\alpha}{\rho f_{\max}}+\frac{1}{B^dr^d}}. \label{19}\\
		\end{aligned}
	\end{equation}
	From \eqref{19}, we observe that $D^{\mathbf{SC}^3}$ is monotonically decreasing with $[\frac{1}{\rho B^u r^u} + \frac{\alpha}{\rho f_{\max}} + \frac{1}{B^d r^d}]$. This expression represents the time to transfer or process one-bit task-related information, where $[\frac{1}{\rho B^u r^u}]$, $[\frac{\alpha}{\rho f_{\max}}]$, and $[\frac{1}{B^d r^d}]$ represent the time for UL, computing, and DL, respectively. Therefore, we move to consider the time-minimization problem as follows:
	\begin{subequations}
		\begin{align}
			\mbox{(PA-3)} \ \ \min\limits_{B^u,B^d} \ \ &\frac{1}{\rho B^ur^u}+\frac{\alpha}{\rho f_{\max}}+\frac{1}{B^dr^d} \\
			\text{s.t.} \ \
			&B^u+B^d\leqslant B_{\max} \label{20b} \\
			&B^u\geqslant 0 \ B^d\geqslant 0.
		\end{align}
	\end{subequations}

	\begin{theorem}
		The optimal UL\&DL bandwidth allocation is given by:
		\begin{subequations}
			\label{opt-B}
			\begin{align}
				(B^u)^*&=\frac{\sqrt{r^d}B_{\max}}{\sqrt{\rho r^u}+\sqrt{r^d}} \\
				(B^d)^*&=\frac{\sqrt{\rho r^u}B_{\max}}{\sqrt{\rho r^u}+\sqrt{r^d}}.
			\end{align}
		\end{subequations}
	The optimal closed-loop information and optimal LQR cost are given by:
		\begin{equation}
				(D^{\mathbf{SC}^3})^*=\frac{T}{\frac{1}{B_{\max}r^{\text{comm}}}+\frac{1}{f_{\max}r^{\text{comp}}}}  \label{22a}\\
		\end{equation}
	\begin{equation}
	\label{f15}
	(l)^*= \left\{
	\begin{aligned}
		&+\infty,  \quad\quad \quad\quad\quad\quad \quad\quad  \quad(D_{\mathbf{SC}^3})^*\leqslant \log_2|\det \mathbf{A}|\\
		 &\frac{n N \!\left( \mathbf{v}\right)|\det \mathbf{M}|^\frac{1}{n}} {2^{\frac{2}{n}\bigg((D_{\mathbf{SC}^3})^*-\log_2|\det \mathbf{A}|\bigg)}-1}+\text{tr}\left( \mathbf{\Sigma}_{\mathbf{v}}\mathbf{S}\right), \text{otherwise,} \\
	\end{aligned}
	\right.
	\end{equation}
	where $r^{\text{comm}}$ denotes the SE of the $\mathbf{SC}^3$ loop, referred to as closed-loop SE, and $r^{\text{comp}}$ denotes the computing efficiency (CE):
	\begin{subequations}
		\label{closed-rate}
		\begin{align}
				r^{\text{comm}}&=\frac{\rho r^ur^d}{(\sqrt{\rho r^u}+\sqrt{r^d})^2} \ \text{(bits/s/Hz)}\label{24a}\\
				r^{\text{comp}}&=\frac{\rho}{\alpha} \ \text{(bits/cycle)}.
		\end{align}
	\end{subequations}
	\end{theorem}
	\begin{proof}
		See Appendix B.
	\end{proof}
	From \eqref{24a}, we can see that the closed-loop SE is jointly determined by the UL\&DL SEs.
	It has the following approximation,
		\begin{equation}
			\label{24}
			r^{\text{comm}}=\left\{
			\begin{aligned}
				&\min(\rho r^u, r^d), \quad \quad\rho r^u>> r^d \ \text{or} \ \rho r^u<< r^d\\
				&\frac{1}{4} \rho r^u,\quad \quad \quad \ \quad \quad r^u\simeq r^d.
			\end{aligned}
			\right.
		\end{equation}
	In fact, $\rho r^u$ and $r^d$ represent the task-level SE of UL\&DL, respectively. From \eqref{24}, we can learn that closed-loop SE is determined by the weak link. A task-level balance of UL\&DL SEs is required to ensure a large closed-loop SE.


	In addition, from \eqref{22a}, we observe an interesting trade-off between bandwidth and CPU frequency for a given closed-loop performance. This trade-off is achieved by adjusting communication and computing time. For instance, increasing the bandwidth reduces communication time, $\frac{1}{(B_{\max}+\Delta B)r^{\text{comm}}}$, which in turn allows for an increase in computing time, $\frac{1}{(f_{\max}-\Delta f)r^{\text{comp}}}$, leading to a reduction of the CPU frequency. Consequently, the relationship between $\Delta B$ and $\Delta f$ can be derived to illustrate how bandwidth and CPU frequency can be exchanged for a given closed-loop performance:
	\begin{equation}
		\label{interchange}
		\begin{aligned}
			&\frac{1}{r^{\text{comm}}(B_{\max}+\Delta B)}+\frac{1}{r^{\text{comp}}(f_{\max}-\Delta f)}\\
			&\quad \quad \quad \quad \quad \quad \quad \quad =\frac{1}{r^{\text{comm}}B_{\max}}+\frac{1}{r^{\text{comp}}f_{\max}},\\
			\Rightarrow \ \ &\Delta B=\frac{B_{\max}}{\frac{r^{\text{comp}}}{r^{\text{comm}}}(\frac{f_{\max}^2}{\Delta fB_{\max}}-\frac{f_{\max}}{B_{\max}})-1}.
		\end{aligned}
	\end{equation}
From the above expression, we observe that $\Delta B$ is influenced by the working point of the $\mathbf{SC}^3$ loop, characterized by the bandwidth, CPU frequency, and SE-to-CE ratio, $(B_{\max}, f_{\max}, \frac{r^{\text{comm}}}{r^{\text{comp}}})$. With a fixed $\Delta f$, a large bandwidth is required when $B_{\max}$ is high, $f_{\max}$ is low, and $\frac{r^{\text{comm}}}{r^{\text{comp}}}$ is high, which corresponds to the communication-saturated region.

	\subsection{Multi-$\mathbf{SC}^3$-Loop Optimization}
	Based on the optimal intra-loop configuration, we move to the inter-loop allocation  problem. Using the optimal LQR cost expression \eqref{f15}, (P1) is simplified into the following inter-loop bandwidth and CPU frequency allocation problem:
	\begin{subequations}
	\begin{align}
		\mbox{(PB-1)} \ \ &\min\limits_{\mathcal{B},\mathcal{F}} \sum_{k=1}^Kl_k \\
		\text{s.t.} \ &l_k \geqslant \frac{n N \!\left( \mathbf{v}_k\right)|\det \mathbf{M}_k|^\frac{1}{n}} {2^{\frac{2}{n}(D^{\mathbf{SC}^3}_k-\log_2|\det \mathbf{A}_k|)}-1}+\text{tr}\left( \mathbf{\Sigma}_{\mathbf{v}_k}\mathbf{S}_k\right), \forall k \label{26b} \\
		&D^{\mathbf{SC}^3}_k\leqslant \frac{T_k}{\frac{1}{ B_kr_k^{\text{comm}}}+\frac{1}{f_k r_k^{\text{comp}}}}, \forall k \label{2-SC3}\\
		&D^{\mathbf{SC}^3}_k> \log_2|\det \mathbf{A}_k|,  \forall k \label{27d}\\
		&\sum\limits_{k=1}^KB_k\leqslant B_{\max}\\
		&\sum\limits_{k=1}^Kf_k\leqslant f_{\max} \\
		&B_k\geqslant 0 \ f_k\geqslant 0, \forall k, \label{27g}
	\end{align}
\end{subequations}
where $B_k$ denotes the bandwidth allocated to the $k$-th $\mathbf{SC}^3$ loop and we denote $\mathcal{B}=\{B_k\}_{k=1}^K$. By calculating the second-order derivative of the  right-hand side of \eqref{26b}, it is easy to prove that it is a convex expression of $D_k^{\mathbf{SC}^3}$. Thus, the only non-convex constraint in (PB-1) is \eqref{2-SC3}. To address this, we first transform it into
	\begin{equation}
	\frac{\frac{1}{ B_kr_k^{\text{comm}}}+\frac{1}{f_k r_k^{\text{comp}}}}{T_k}\leqslant \frac{1}{D^{\mathbf{SC}^3}_k} \label{48}.
	\end{equation}
	Since $[\frac{1}{x}]$ is a convex expression of $x$ when $x>0$, we can see that only the right-hand side of \eqref{48} does not satisfy the rule of convex optimization.  By using the successive convex approximation, we propose an iterative algorithm. The convex optimization problem in the $s$-th iteration is given by
		\begin{subequations}
			\label{P3}
		\begin{align}
			\mbox{(PB-2)}& \ \min\limits_{\mathcal{B},\mathcal{F}} \sum_{k=1}^Kl_k \\
			\text{s.t.} \ &l_k \geqslant \frac{n N \!\left( \mathbf{v}_k\right)|\det \mathbf{M}_k|^\frac{1}{n}} {2^{\frac{2}{n}(D^{\mathbf{SC}^3}_k-\log_2|\det \mathbf{A}_k|)}-1}+\text{tr}\left( \mathbf{\Sigma}_{\mathbf{v}_k}\mathbf{S}_k\right), \forall k \label{2-a1} \\
			&\frac{\frac{1}{ B_kr_k^{\text{comm}}}+\frac{1}{f_k r_k^{\text{comp}}}}{T_k}	\leqslant \frac{2(D^{\mathbf{SC}^3}_k)^{s-1}-D^{\mathbf{SC}^3}_k}{\big((D^{\mathbf{SC}^3}_k)^{s-1}\big)^2},  \forall k \label{iterative}\\
			&\eqref{27d}-\eqref{27g}, \nonumber
		\end{align}
	\end{subequations}
	where the right side of \eqref{iterative} is the Taylor expansion of  $\big[\frac{1}{D^{\mathbf{SC}^3}_k}\big]$ at the result obtained in the $(s-1)$-th iteration, $(D^{\mathbf{SC}^3}_k)^{s-1}$.
	We summarize the proposed iterative algorithm in {\bf{Algorithm 1}}. The convergence of the algorithm is proved in Appendix C.
		 \begin{algorithm}[t]
		\caption{Iterative Algorithm for Goal-Oriented Closed-Loop Optimization Scheme}  \small
		\label{Tab1}
		\begin{algorithmic}[1]
			\REQUIRE
			{The number of the $\mathbf{SC}^3$ loops, $K$, and the iteration terminating threshold, $\delta$. \\
			Control related parameters: $n$, $\log_2|\det\mathbf{A}_k|$,  $\mathbf{B}_k$, $\mathbf{Q}_k$, $\mathbf{R}_k$, $T_k$, and $\Sigma_{\mathbf{v}_k}$;\\
			Communication related parameters: $r^u_k$, $r^d_k$, and $B_{\max}$;\\
			Computing related parameters: $\alpha_k$, $\rho_k$,  $f_{\max}$;}
			\STATE  Calculate $\mathbf{S}_k$ and $\mathbf{M}_k$ according to \eqref{Riccati};
			\STATE Calculate the closed-loop SE and CE  for each $\mathbf{SC}^3$ loop according to \eqref{closed-rate};
			\STATE \emph{Initialization}: $s=0$ and $(D^{\mathbf{SC}^3}_k)^{0}=\log_2|\det\mathbf{A}_k|+1, \forall k$;
			\REPEAT
			\STATE  $s=s+1$;
			\STATE  Solve (PB-2) to obtain $(B_k)^s$, $(f_k)^s$, and $(D^{\mathbf{SC}^3}_k)^s$;
			\UNTIL{$\frac{\abs{\sum\limits_{k=1}^K(l_k)^{s}-\sum\limits_{k=1}^K(l_k)^{s-1}}}{\sum\limits_{k=1}^K(l_k)^{s-1}}\leqslant \delta$ for (P3)}
			\STATE Calculate the optimal UL\&DL bandwidth allocation for each $\mathbf{SC}^3$ loop according to \eqref{opt-B};
			\STATE  Calculate the optimal time allocation for each $\mathbf{SC}^3$ loop according to \eqref{opt};
			\ENSURE The UL\&DL bandwidth,  time, and CPU frequency of $K$ $\mathbf{SC}^3$ loops: $(\mathcal{B}^{u})^*$, $(\mathcal{B}^{d})^*$, $(\mathcal{T}^{u})^*$, $(\mathcal{T}^{d})^*$, and $(\mathcal{F})^*$.
				\end{algorithmic}
	\end{algorithm}

Next, we derive the approximate closed-form solution for the inter-loop bandwidth allocation. We assume that the CPU frequency is sufficiently adequate such that the computing time becomes negligible compared to the communication time, $\frac{1}{B_k r_k^{\text{comm}}} + \frac{1}{f_k r_k^{\text{comp}}} \approx \frac{1}{B_k r_k^{\text{comm}}}$. On this basis, the LQR cost can be approximated as follows:
\begin{equation}
	\label{approx}
	l_k \approx \frac{n N \!\left( \mathbf{v}_k\right)|\det \mathbf{M}_k|^\frac{1}{n}} {2^{\frac{2}{n}(T_kB_kr_k^{\text{comm}}-\log_2|\det \mathbf{A}_k|)}}+\text{tr}\left( \mathbf{\Sigma}_{\mathbf{v}_k}\mathbf{S}_k\right),
\end{equation}
where the term [$-1$] in the denominator is omitted from its original expression \eqref{f15}, under the assumption that the system operates in the assured-to-be-stable region, i.e., $T_k B_k r_k^{\text{comm}} \gg \log_2|\det \mathbf{A}_k|$. Using \eqref{approx}, (PB-1) is simplified into a convex bandwidth allocation problem,
	\begin{subequations}
	\begin{align}
		\mbox{(PB-3)} \  \min\limits_{\mathcal{B}} \ &\sum_{k=1}^K\frac{n N \!\left( \mathbf{v}_k\right)|\det \mathbf{M}_k|^\frac{1}{n}} {2^{\frac{2}{n}(T_kB_kr_k^{comm}-\log_2|\det \mathbf{A}_k|)}}\\
		\text{s.t.}
		&\sum\limits_{k=1}^KB_k\leqslant B_{\max}  \label{B24}\\
		&B_k\geqslant 0,\forall k,
	\end{align}
\end{subequations}
where $\text{tr}\left( \mathbf{\Sigma}_{\mathbf{v}_k}\mathbf{S}_k\right)$ is omitted as it does not influence the allocation results. The stable condition is also omitted as in the assure-to-be-stable region. 
\begin{theorem}
The optimal bandwidth allocation to (PB-3) is given by
\begin{equation}
	\label{q1}
	\begin{aligned}
	&(B_k)^*=\frac{n}{2r^{\text{comm}}_kT_k} \times \\
	&\bigg(\frac{\sum\limits_{i=1,i\neq k}^{K}\frac{n}{2r^{\text{comm}}_iT_i}[(e_k-e_i)+\log_2(\frac{r^{\text{comm}}_kT_k}{r^{\text{comm}}_iT_i})]+B_{\max}}{\sum\limits_{i=1}^{K}\frac{n}{2r^{\text{comm}}_iT_i}}\bigg),
	\end{aligned}
\end{equation}
where $e_k$ denotes the control-related parameter,
\begin{equation}
	e_k\triangleq \log_2(N(\mathbf{v}_k))+\frac{2}{n}\log_2|\det\mathbf{A}_k|.
\end{equation}
\end{theorem}
\begin{proof}
	See Appendix D.
\end{proof}
We can learn from \eqref{q1} that bandwidth allocation is influenced by both the communication parameter, $r_k^{\text{comm}}$, and the control parameter, $e_k$. The $\mathbf{SC}^3$ loop with poor communication performance (smaller $r^{\text{comm}}$) or controlling a more unstable system (larger $e_k$) is allocated more bandwidth. This shows that this inter-loop bandwidth allocation is fairness-minded, which balances the control processes across different $\mathbf{SC}^3$ loops. 

\begin{figure*} [t]
	\centering
	\includegraphics[width=0.8\linewidth]{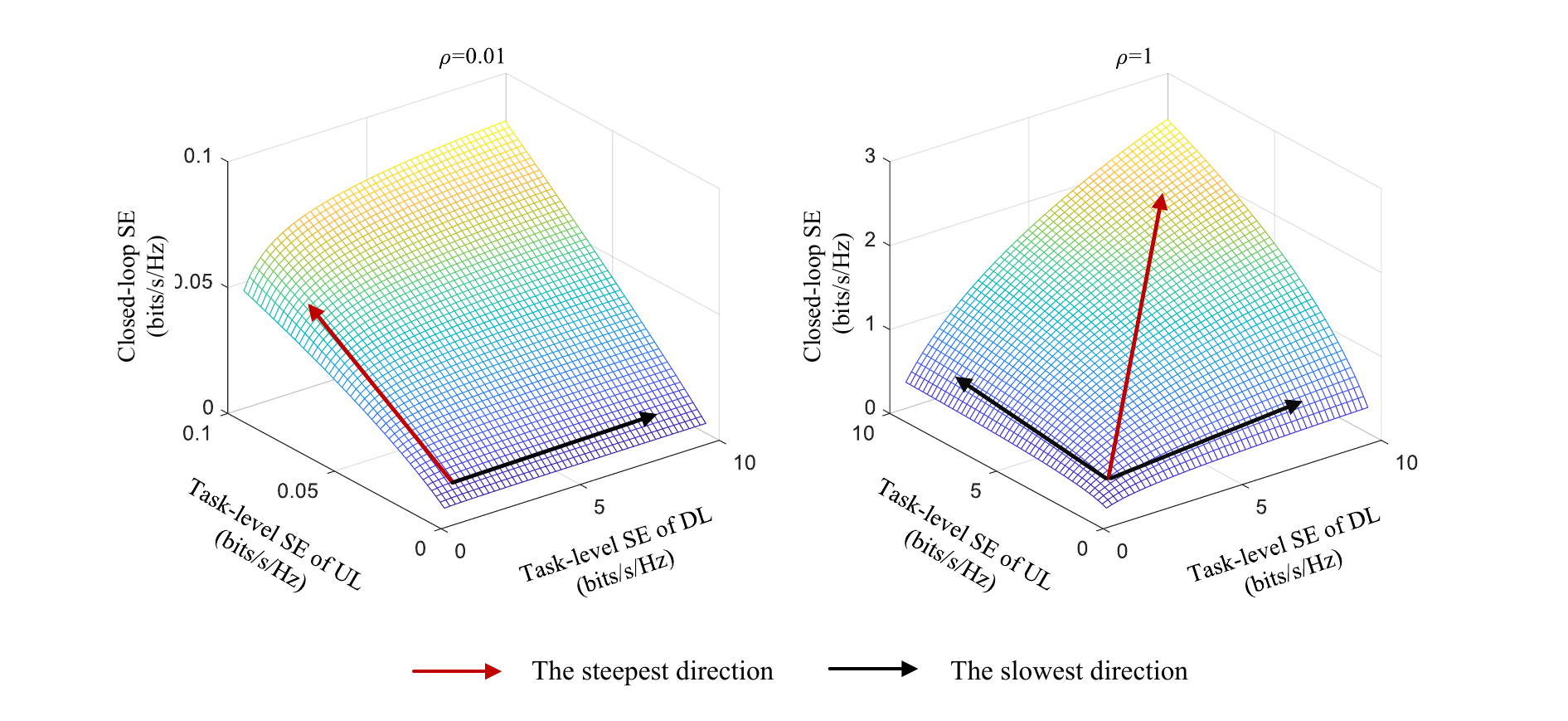}
	\caption{The closed-loop SE varying with the task-level SEs of UL\&DL. The left subfigure illustrates the case of imbalanced task-level SEs between UL\&DL, and the right figure illustrates the balanced case. The red arrow and black arrow indicate the steepest and slowest directions to improve the closed-loop SE. }
	\label{sim1}
\end{figure*}
\section{Simulation Results and Discussion}
\label{section 4}
In this section, we present the simulation results and discussion. We consider an unmanned robotic system composed of $K=4$ $\mathbf{SC}^3$ loops. The channel noise variance is $\sigma^2 = -107$ dBm \cite{fang2}, and the large-scale fading coefficients, $g_u$ and $g_d$, are calculated using the path-loss model given in \eqref{channel_gain}, where $f_c = 2$ GHz, and UL\&DL distances are set as $d^u_k = [1, 1.5, 2, 2.5]$ km  and $d^d_k = [3, 3.5, 4, 5]$ km, respectively. The small-scale fading is randomly generated following $\mathcal{CN}(0,1)$ \cite{fang2}. The UL\&DL transmit power is adjusted to maintain a constant received SNR, i.e., $\frac{|h^{u}_k|^2 p^{u}_k}{\sigma^2} = \frac{g^u_k}{\sigma^2}$ and $\frac{|h^{d}_k|^2 p^{d}_k}{\sigma^2} = \frac{g^d_k}{\sigma^2}$. Accordingly, we can calculate the UL\&DL SEs by \eqref{speed}, i.e,  $r^u_k = [10.5, 9.9, 9.5, 9.2]$ bits/s/Hz and $r^d_k = [12.2, 12.0, 11.8, 11.6]$ bits/s/Hz.
The computing-related parameters are set as $f_{\max} = 2$ GHz, $\rho_k = 0.01$, $\forall k$, and $\alpha_k = [100, 200, 1000, 50]$ (cycles/bit). The control-related parameters are configured as $n= 100$, $\log|\det \mathbf{A}_k| = [10, 20, 30, 40]$, $ \mathbf{R} = \mathbf{0}_{100}$, $\mathbf{Q} = \mathbf{I}_{100}$ \cite{Fang}, and $T_k=T= 10$ ms, $\forall k$. The iteration termination threshold is set as $\delta = 0.001$, and the optimization tool is CVX \cite{CVX}. Unless otherwise stated, these parameters are used throughout the results below.

\subsection{Single-$\mathbf{SC}^3$-Loop Simulation}
In Fig. \ref{sim1}, we present the relationships between UL\&DL SEs and the closed-loop SE, as described in \eqref{24}.
In the left figure, we set the information extraction ratio as $\rho=0.01$, which creates the unbalance of UL\&DL SEs at the task level. In this case, the closed-loop SE is constrained by the UL.  As indicated by the red and black arrows, the  most effective way to enhance the closed-loop SE is by improving the UL SE, while improving the DL SE provides little improvement.
In contrast, the right figure illustrates the case where the information extraction ratio is set as $\rho=1$, resulting in relatively balanced task-level UL\&DL SEs.  In this case, the most effective way to enhance the closed-loop SE is to simultaneously improve UL\&DL SEs.
It is noted that, in practice, the raw data transmitted via the UL often contain redundant and irrelevant information, leading the information extraction ratio a very small number, $\rho\ll 1$. An enhanced UL is usually required for keeping the task-level balance between UL\&DL.

In Fig. \ref{sim2}, we present the LQR cost under different UL\&DL configurations. Here, we compare the proposed scheme with two classical schemes: equal allocation and proportional allocation.
\begin{itemize}
	\item Equal allocation: The bandwidth is equally divided for UL\&DL like FDD, $B^d=B^u=\frac{B_{\max}}{2}$.
	\item Proportional allocation: The UL\&DL bandwidth is proportional allocated per the information extraction ratio, $B^u=\frac{1}{1+\rho}B_{\max}$ and $B^d=\frac{\rho}{1+\rho}B_{\max}$.
\end{itemize}
In above two schemes, the DL transmission time is set as $t^d=\frac{T}{3}$, and the remaining time is adjusted between UL and computing by letting the UL transmitted information rightly processed by computing, i.e., $t^u+ \frac{\alpha D^u}{f_{\max}}=\frac{2T}{3}$. The maximal bandwidth is set as $B_{\max}=500$ kHz, and the CPU frequency is set as $f_{\max}=0.5$ GHz. Fig. \ref{sim2} is a dual-axis chart with the left side representing the task-related information and the right side representing the LQR cost. From the bar, we can see that, under the proposed scheme, the UL\&DL are aligned to transmit the same amount of task-related information, while the other two schemes fail to achieve this balance. Their weak ULs constrain the control performance of the $\mathbf{SC}^3$ loop, leading to the high LQR cost (shown by the black curve). In addition, we can observe that the LQR cost is higher with more degree of UL\&DL imbalance, indicating the importance of keeping task-level UL\&DL balance within the $\mathbf{SC}^3$ loop.  

\begin{figure} [t]
	\centering
	\includegraphics[width=0.9\linewidth]{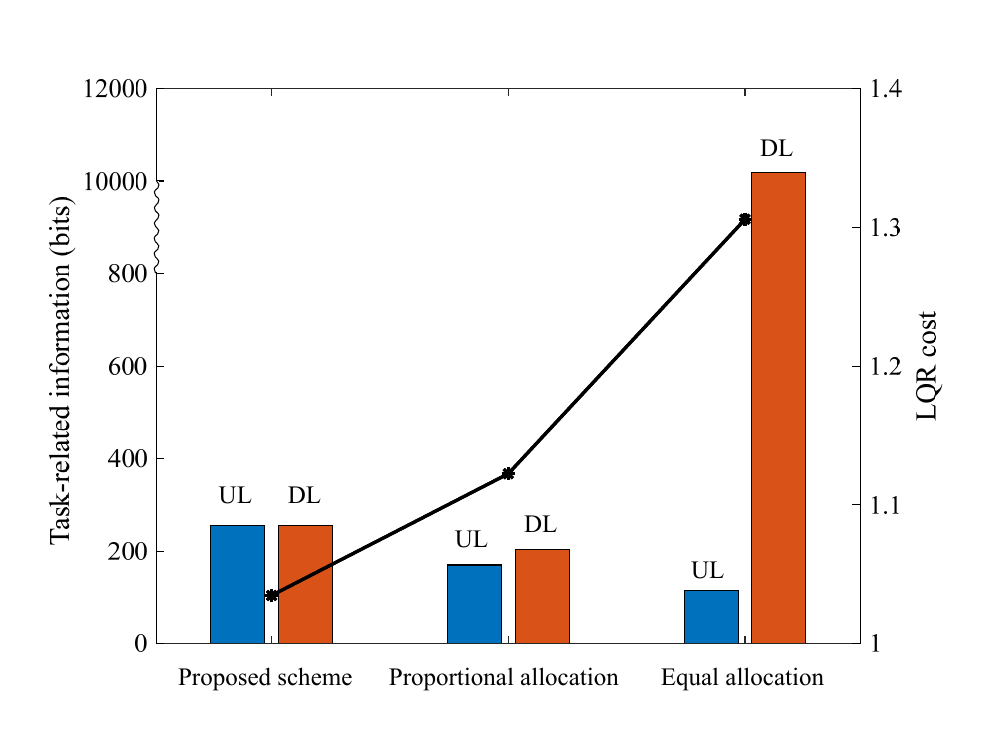}
	\caption{The task-related information and LQR cost under three UL\&DL configuration schemes.}
	\label{sim2}
\end{figure}

In Fig. \ref{simm}, we illustrate the interchange relationships between bandwidth and CPU frequency. This figure shows the required bandwidth in exchange for 1 MHz CPU frequency under different working points of the $\mathbf{SC}^3$ loop, described by the  bandwidth, CPU frequency, and SE-to-CE ratio, i.e., ($B$, $f$,$\frac{r^{\text{comm}}}{r^{\text{comp}}}$). The color gradient represents the amount of required bandwidth, increasing as the color shifts from blue to yellow.  The figure shows that in the communication-saturated region, more bandwidth is required compared to the communication-limited region. For example, at the working point of $(1 \text{MHz}, 2 \text{GHz}, 5 \times 10^4)$, the required bandwidth is $\Delta B = 12.6 \ \text{kHz}$, whereas at $(2 \text{MHz}, 1 \text{GHz},  10^5)$, the required bandwidth is $\Delta B = 500 \ \text{kHz}$, with an increase of forty-fold. This indicates the marginal effect, that the effectiveness of additional bandwidth diminishes as the $\mathbf{SC}^3$ loop approaches the communication-saturated region.
\begin{figure} [t]
	\centering
	\includegraphics[width=0.9\linewidth]{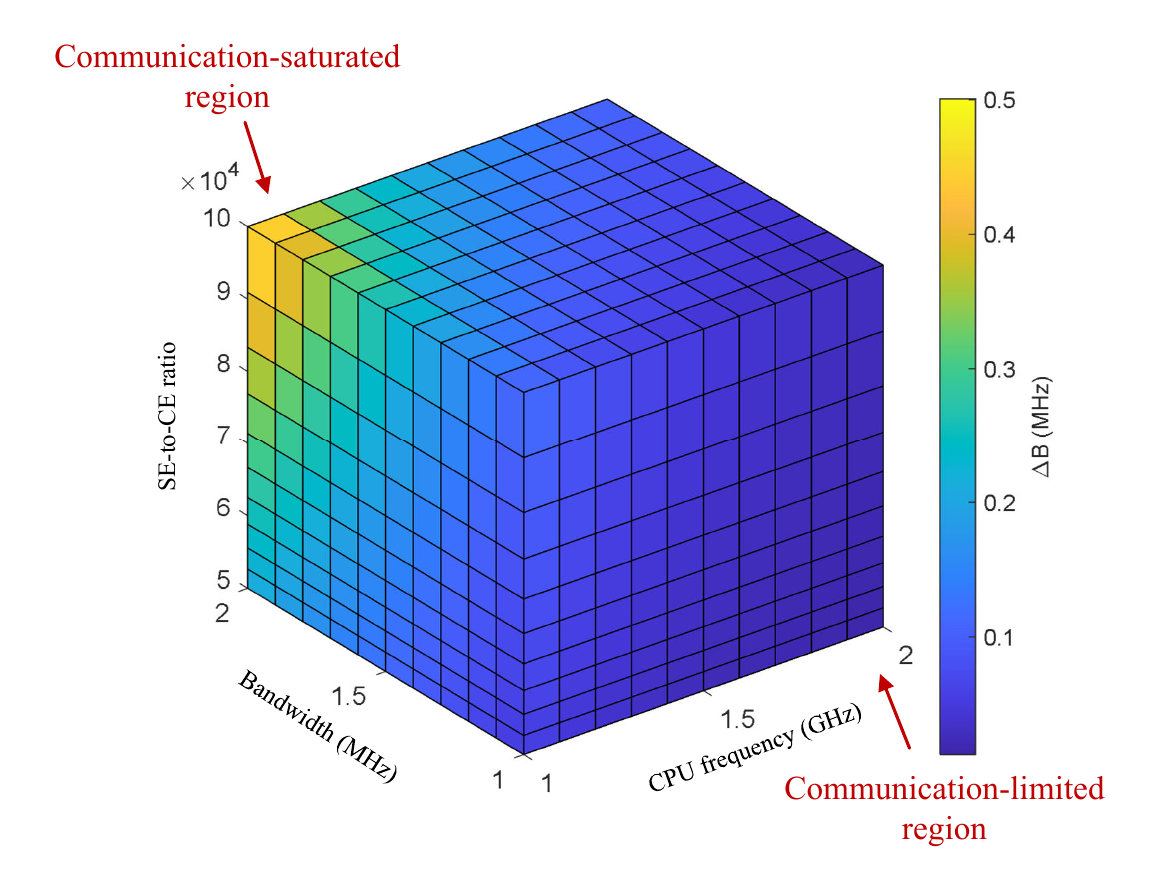}
	\caption{The required communication bandwidth to exchange $1$  MHz computing CPU frequency varies with ($B$,$f$,$\frac{r^{\text{comm}}}{r^{\text{comp}}}$).}
		\label{simm}
\end{figure}

\subsection{Multi-$\mathbf{SC}^3$-Loop Simulation}

\begin{figure} [t]
	\centering
	\includegraphics[width=0.9\linewidth]{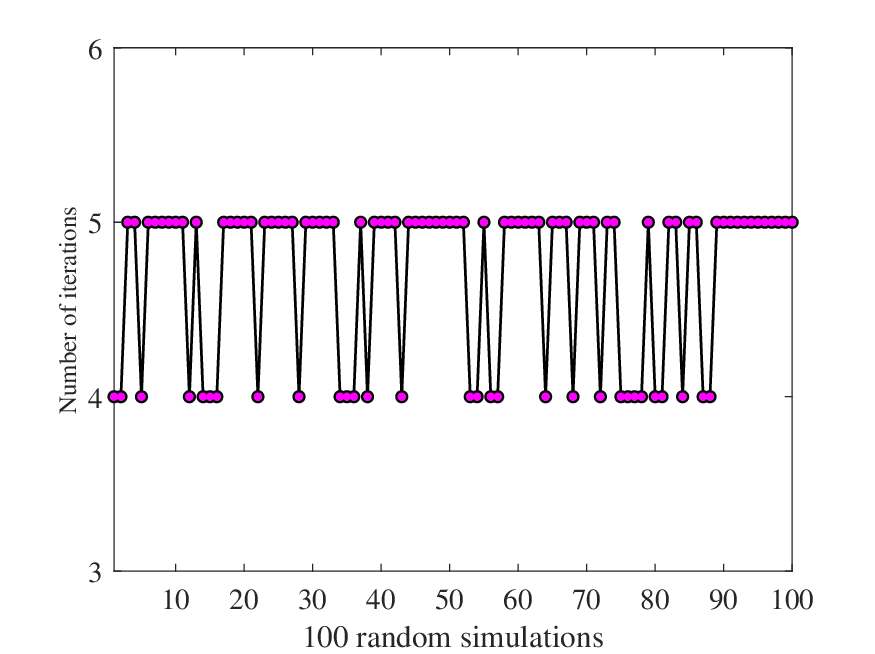}
	\caption{The number of iterations of the proposed iterative algorithm under 100 random simulations.}
	\label{fig2}
\end{figure}
In Fig. \ref{fig2}, we provide the convergence performance of the proposed iterative algorithm. In this simulation, we randomly generate UL\&DL distances within the range of $[0.5, 5]$ km, the maximum bandwidth is set as $B_{\max} = 1$ MHz. The results show that, across all 100 simulations, the  proposed algorithm consistently requires only 4 or 5 iterations to converge. 

\begin{figure} [t]
	\centering
	\includegraphics[width=0.9\linewidth]{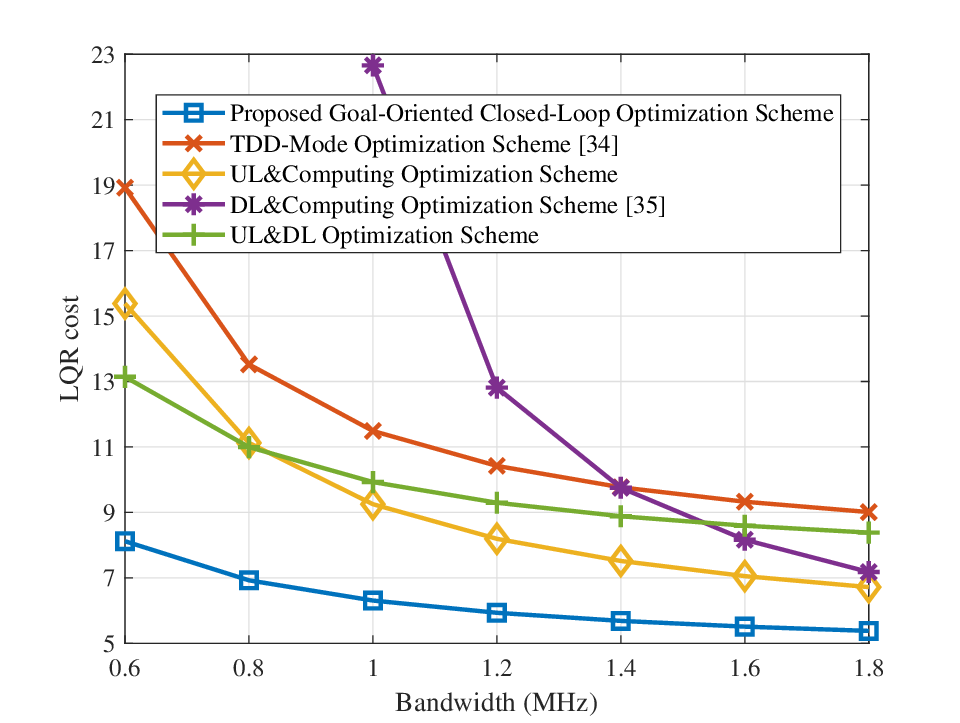}
	\caption{The LQR cost varies with the bandwidth under different allocation schemes.}
	\label{fig6}
\end{figure}
In Fig. \ref{fig6} we compare the proposed scheme with a TDD-mode optimization scheme \cite{Wen2} and three separate schemes, with each focused on a part of the $\mathbf{SC}^3$ loop. For fair comparisons, we consider these schemes to be goal-oriented by taking the LQR cost as the objective.
\begin{itemize}
	\item TDD-Mode Optimization Scheme \cite{Wen2}: This scheme uses TDD to access four $\mathbf{SC}^3$ loops. Each loop is allocated with an equal time slot,  $T_k=\frac{T}{4}=2.5$ ms, and exclusively uses bandwidth and CPU frequency resources in this slot. It optimizes the UL\&DL, and computing time for each $\mathbf{SC}^3$ loop.
	\item UL\&Computing Optimization Scheme: This scheme optimizes UL transmission and computing, while the DL parameters are fixed as $B^d_k=\frac{B_{\max}}{2K}, \forall k$ and $t^d_k=1 \ \text{ms}, \forall k$.
	\item DL\&Computing Optimization Scheme \cite{Jian}: This scheme
	optimizes DL transmission and computing, with the UL parameters fixed at $B^u_k=\frac{B_{\max}}{2K}, \forall k$ and $t^u_k=4 \text{ms}, \forall k$.
	\item UL\&DL Optimization Scheme: This scheme optimizes UL\&DL transmission, with the CPU frequency fixed at $f_k=\frac{f_{\max}}{K}, \forall k$.
\end{itemize}
It can be observed that the proposed scheme achieves the lowest LQR cost compared to the other four schemes.
In the bandwidth-limited region, the DL\&Computing Optimization Scheme \cite{Jian} performs the worst. This is because the fixed UL is the bottleneck, which largely determines the $\mathbf{SC}^3$-loop performance: the data transmitted via the UL are the data processed by computing, which in turn determines the accuracy of the command transmitted via the DL. Conversely, in the bandwidth-adequate region, the TDD-Mode Optimization Scheme and UL\&DL Optimization Scheme perform poorly. This is because these two schemes constrain the inter-loop adjustment, negatively affecting the balance of the control progress across different $\mathbf{SC}^3$ loops. This highlights the advantage of treating the $\mathbf{SC}^3$ loop as an integrated structure and jointly configuring resources within and across the $\mathbf{SC}^3$ loops.

\begin{figure} [t]
	\centering
	\includegraphics[width=0.9\linewidth]{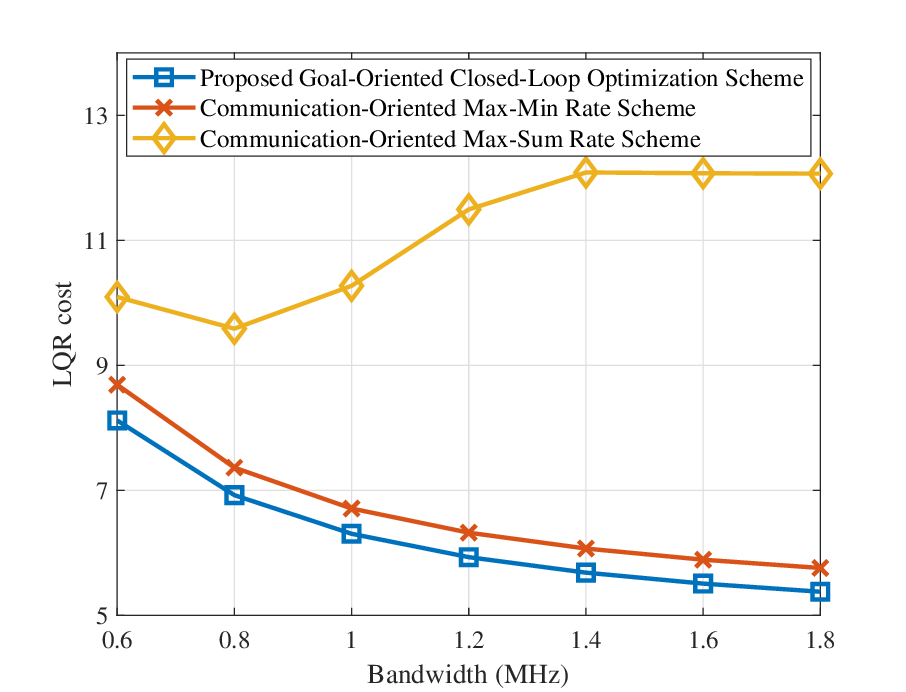}
	\caption{The LQR cost varies with the bandwidth under the goal-oriented and communication-oriented schemes.}
	\label{fig3}
\end{figure}
In Fig. \ref{fig3}, we compare the proposed goal-oriented scheme with two communication-oriented schemes. 
\begin{itemize}
	\item Communication-Oriented Max-Sum Rate Scheme: This scheme maximizes the sum of the closed-loop information, i.e., $\max\sum\limits_{k=1}^K{D^{\mathbf{SC}^3}_k}$.
	\item Communication-Oriented Max-Min Rate Scheme: This scheme maximizes the minimal  closed-loop information, i.e., $\max\min\limits_{k}{D^{\mathbf{SC}^3}_k}$.
\end{itemize}
These two communication-oriented schemes can be solved using a similar iterative algorithm by replacing the objective of the LQR cost  with the sum-rate and min-rate objective, i.e., $\max\sum\limits_{k=1}^K D^{\mathbf{SC}^3}_k$ and $\max\min\limits_{k} D^{\mathbf{SC}^3}_k$, and convert the LQR constraint \eqref{2-a1} into a closed-loop-information constraint,
\begin{equation}
	\label{sim_con}
	D^{\mathbf{SC}^3}_k\geqslant \frac{n}{2}\log_2(\frac{n N \!\left( \mathbf{v_k}\right)|\det \mathbf{M}_k|^\frac{1}{n}}{l_k-\text{tr}(\mathbf{\mathbf{\Sigma}_k\mathbf{S}_k})}+1)+\log_2|\det\mathbf{A}_k|.
\end{equation}
In this simulation, we set the LQR cost requirement in \eqref{sim_con} as $l_k=5, \forall k$. From the figure we can see that, the Communication-Oriented Max-Sum Rate Scheme performs the worst. 
This is because this scheme disproportionately allocates the most resources to the $\mathbf{SC}^3$ loops with the highest SE and CE, amplifying the imbalance among different $\mathbf{SC}^3$ loops.
As the bandwidth increases, this bias intensifies, leading the LQR cost of other $\mathbf{SC}^3$ loops to rise sharply, resulting in the increase of the sum LQR cost.
In addition, we can observe that, when $B_{\max}\geqslant 1.4$ MHz, the LQR cost reaches its maximum and remains unchanged under the Communication-Oriented Max-Sum Rate Scheme. In this region, the control performance of the  $\mathbf{SC}^3$ loops diverges into two extremes: the LQR cost of the $\mathbf{SC}^3$ loops with the highest SE or CE approaches its minimal value, while the LQR cost of the remaining $\mathbf{SC}^3$ loops barely meets the basic requirement as given in \eqref{sim_con}.  Conversely, the Communication-Oriented Max-Min Rate Scheme performs similarly to the proposed scheme. They have a small performance gap, which is because the Max-Min Rate Scheme prioritizes the rate balance while overlooking the differences of controlled systems. The good performance of the Max-Min Rate Scheme and the proposed scheme indicates the importance of keeping balance across different $\mathbf{SC}^3$ loops.

\begin{figure} [t]
	\centering
	\includegraphics[width=0.9\linewidth]{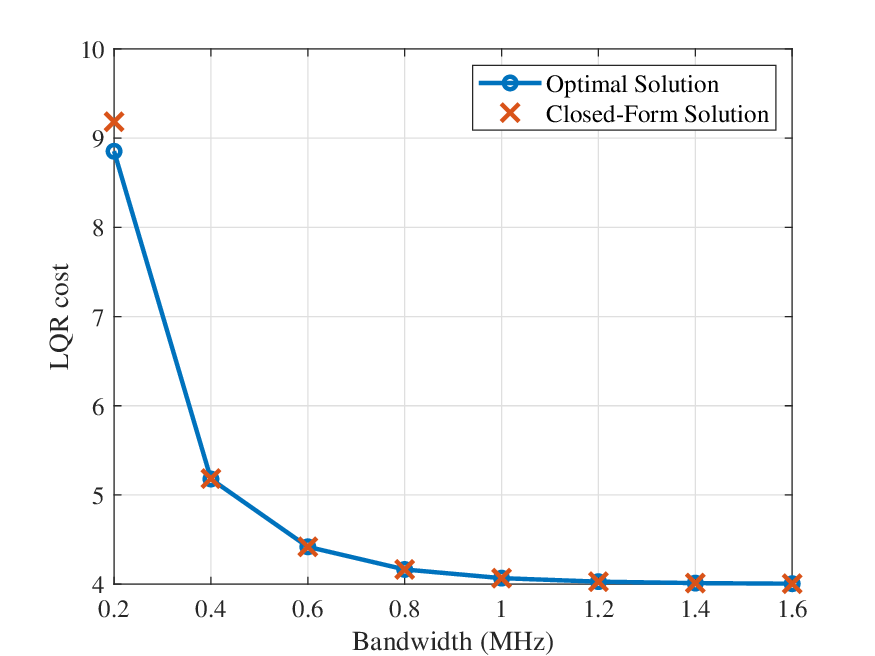}
	\caption{The LQR cost under the approximate closed-form solution given in \eqref{q1} and the optimal solution solved by CVX.}
	\label{fig4}
\end{figure}

Fig. \ref{fig4} compares the approximate closed-form solution for the inter-loop bandwidth allocation \eqref{q1} with the optimal solution obtained by CVX. We can see that the LQR costs under the approximate closed-form solution and the optimal solution exhibit a small gap in the bandwidth-limited region, and this gap rapidly approaches zero as the bandwidth increases. This demonstrates the accuracy of the approximate closed-form solution, particularly when the system has sufficient communication and computing resources.

Furthermore,  we show the bandwidth allocation principles by comparing the proposed goal-oriented scheme with two communication-oriented schemes in Fig. \ref{fig5}. The maximal bandwidth is set as $B_{\max}=1$ MHz. In the top subfigure, we set the four $\mathbf{SC}^3$ loops with different intrinsic entropy, $\log_2|\det\mathbf{A}_k| = [10, 20, 100, 200]$, and equal SE, $r^{\text{comm}}_k = 0.1$ bits/s/Hz, $\forall k$. It can be seen that the proposed scheme allocates more bandwidth to the $\mathbf{SC}^3$ loops with higher intrinsic entropy, whereas two communication-oriented schemes allocate bandwidth equally. This is because the communication-oriented schemes focus on the data rate, and as a result, they are unable to perceive the stability differences of the controlled systems. 
In the bottom subfigure, we set the four $\mathbf{SC}^3$ loops with different SEs, $r^{\text{comm}}_k=[0.08, 0.10, 0.12,0.14]$ bits/s/Hz, and equal intrinsic entropy, $\log_2|\det\mathbf{A}_k|=20, \ \forall k$. In this case, the Max-Sum Rate Scheme disproportionately allocates the most resources to the $\mathbf{SC}^3$ loop with the highest SE (loop 4), leaving other $\mathbf{SC}^3$ loops with minimal resources to satisfy the LQR-cost constraint \eqref{sim_con}. In contrast, both the proposed scheme and the Max-Min Rate Scheme allocate more bandwidth to the $\mathbf{SC}^3$ loops with lower SEs, aiming to achieve the task-level and rate-level balance, respectively.
This shows that the proposed scheme is fairness-minded. Instead of over-resourcing the strong $\mathbf{SC}^3$ loop at the expense of others, it maintains the task-level balance across $\mathbf{SC}^3$ loops to ensure the control performance of the whole unmanned robotic system.

\begin{figure} [t]
	\centering
	\includegraphics[width=1\linewidth]{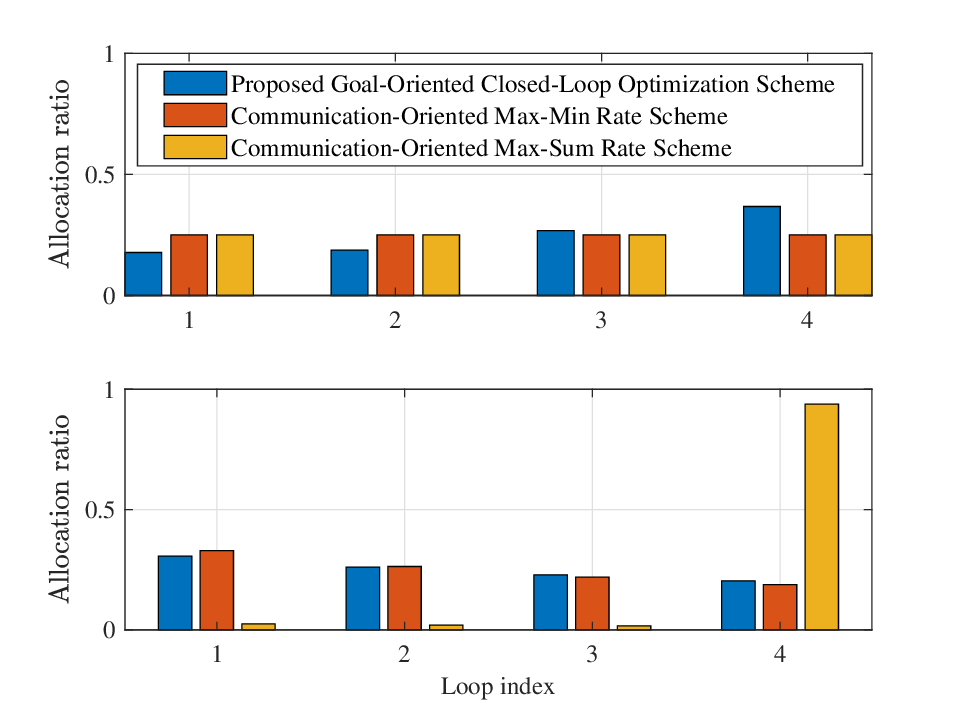}
	\caption{The bandwidth allocation ratio under three different schemes. The top subfigure shows the bandwidth allocation ratio under different intrinsic entropy and the bottom subfigure shows the bandwidth allocation ratio under different closed-loop SEs.}
	\label{fig5}
\end{figure}

%
%
%
%
%
\section{Conclusions}
\label{section 5}
In this paper, we have investigated an unmanned robotic system formed by the reflex-arc-like $\mathbf{SC}^3$ loops. We have proposed a goal-oriented closed-loop optimization scheme that jointly configures UL\&DL bandwidth, time, and CPU frequency to improve the control performance of the system. We have derived closed-form solutions for the UL\&DL bandwidth and time allocation within the $\mathbf{SC}^3$ loop and proposed an iterative algorithm to optimize the inter-loop bandwidth and CPU frequency allocation. Under the condition of adequate CPU frequency availability, we have also derived the approximate closed-form solution for the inter-loop bandwidth allocation.
We have shown that the superiority of the proposed scheme lies in achieving a two-tier task-level balance within and across the $\mathbf{SC}^3$ loops. We believe that such a structured design—taking the $\mathbf{SC}^3$ loop as an integrated structure—will promote the development of unmanned robotic systems in 6G.
\appendices

\section{Proof of \bf{Lemma 1}}
In (PA-2),  the optimal values of $D^{\mathbf{SC}^3}$, $D^u$, and $D^d$ are their maximum value allowed. As a result, the equality in constraints \eqref{16b}--\eqref{16f} must hold at the optimal solution,
\begin{equation}
(D^{\mathbf{SC}^3})^*=\min\{\rho (D^u)^*,(D^d)^*\}
\end{equation}
where
\begin{equation}
	(D^u)^*=t^uB^ur^u, \
	(D^d)^*=t^dB^dr^d.
\end{equation}
If $\rho (D^u)^* \neq (D^d)^*$ at the optimal solution, taking $\rho (D^u)^* > (D^d)^*$ as an example, we thus have $(D^{\mathbf{SC}^3})^*=\min\{\rho (D^u)^*,(D^d)^*\}=(D^d)^*$. In this case, $(D^d)^*$ can be improved by reallocating UL time to the DL as long as $\rho (D^u)^* > (D^d)^*$. This improves the objective by
\begin{equation}
t^d\uparrow \ \rightarrow  (D^d)^*\uparrow \rightarrow   (D^{\mathbf{SC}^3})^*\uparrow,
\end{equation}
which contradicts the assumption that $(D^{\mathbf{SC}^3})^*$ is optimal. Similarly, if $\rho (D^u)^* < (D^d)^*$, we can use the similar reasoning to show that $(D^{\mathbf{SC}^3})^*$ can be improved. Therefore, neither $\rho (D^u)^* > (D^d)^*$ nor $\rho (D^u)^* < (D^d)^*$ hold, meaning that at the optimal solution $\rho (D^u)^*= (D^d)^*$.

Furthermore, based on $\rho (D^u)^*= (D^d)^*$, we have $\rho t^uB^ur^u= t^dB^dr^d$.
The optimal time for DL transmission and computing can be expressed as a function of the optimal time for UL transmission,
\begin{equation}
	\label{clc2}
	(t^d)^*=\frac{\rho (t^u)^*B^ur^u}{B^dr^d}, \  (t^{\text{comp}})^*=\frac{\alpha (t^u)^*B^ur^u}{f}.
\end{equation}
In addition, since more time enables transmitting and processing more information, thereby increasing $D^{\mathbf{SC}^3}$, the cycle-time constraint \eqref{16e} must be satisfied as an equality at the optimal solution,
\begin{equation}
	\label{clc1}
	(t^u)^*+\frac{\alpha D^u}{f}+(t^d)^*=T.
\end{equation}
Then by substituting \eqref{clc2} into \eqref{clc1}, we obtain the optimal expression of the UL transmission time,
\begin{equation}
	\label{36}
	\begin{aligned}
		&(t^u)^*+\frac{\alpha t^uB^ur^u}{f}+\frac{\rho (t^u)^*B^ur^u}{B^dr^d}=T &\Rightarrow \\
		&(t^u)^*=\frac{\frac{1}{\rho B^ur^u}}{\frac{1}{\rho B^ur^u}+\frac{\alpha}{\rho f}+\frac{1}{B^dr^d}}T.
	\end{aligned}
\end{equation}
Based on $(t^u)^*$, we can further obtain $(t^d)^*$ and $(t^{\text{comp}})^*$ by substituting \eqref{36} into \eqref{clc2}. 

\section{Proof of \bf{Theorem 1}}
	In (PA-3), the objective is monotonically decreasing with the UL\&DL bandwidth. Therefore, the bandwidth constraint \eqref{20b} must be satisfied as an equality at the optimal solution,
	\begin{equation}
	(B^u)^*+(B^d)^*=B_{\max}. \label{39}
	\end{equation}
	We let $(B^d)^*=B_{\max}-(B^u)^*$ and substitute it to the objective. By calculating the first-order derivative of the objective and setting it to zero, we derive the optimal closed-form solution for the bandwidth allocation as follows:
	\begin{equation}
		\begin{aligned}
			&\frac{\partial}{\partial B^u}\bigg[\frac{1}{\rho B^ur^u}+\frac{\alpha}{\rho f_{\max}}+\frac{1}{(B_{\max}-B^u)r^d}\bigg]=0 \\
			\Rightarrow &-\frac{1}{\rho (B^u)^2r^u}+\frac{1}{(B_{\max}-B^u)^2r^d}=0 \\
			\Rightarrow &(B^u)^*=\frac{\sqrt{r^d}B_{\max}}{\sqrt{\rho r^u}+\sqrt{r^d}}.
		\end{aligned}
	\end{equation}
	Then, it is easy to further get $(B^d)^*$ according to \eqref{39}. By substituting  $(B^u)^*$ and $(B^d)^*$ into \eqref{19} and \eqref{f6}, we further get the optimal expression of the closed-form information, $(D^{\mathbf{SC}^3})^*$,  and LQR cost, $(l)^*$.

\section{Proof of the convergence of \bf{Algorithm 1}}
We prove the convergence of the proposed algorithm by proving this optimization process leads to a non-increasing objective. We first prove that the solution in the $(s-1)$-th iteration is a feasible solution in the $s$-th iteration.

From convenience, we introduce a function to express the constraint \eqref{iterative} in (PB-2),
\begin{equation}
G(B_k,f_k,D^{\mathbf{SC}^3}_k)\triangleq \frac{\frac{1}{B_kr^{\text{comm}}}+\frac{1}{ f_kr^{\text{comp}}}}{T_k}-\frac{1}{D^{\mathbf{SC}^3}_k}.
\end{equation}
Denote the results obtained  in the $s$-th iteration as $\{(B_k)^s,(f_k)^s,(D^{\mathbf{SC}^3}_k)^s, (l_k)^s, \forall k\}$. The constraint  \eqref{iterative} in the $s$-th iteration can be expressed as
\begin{equation}
G(B_k,f_k,D^{\mathbf{SC}^3}_k|(D^{\mathbf{SC}^3}_k)^{s-1})\geqslant 0,
\end{equation}
where $G(B_k,f_k,D^{\mathbf{SC}^3}_k|(D^{\mathbf{SC}^3}_k)^{s-1})$ represents the Taylor expansion of $G(B_k,f_k,D^{\mathbf{SC}^3}_k)$ at $(D^{\mathbf{SC}^3}_k)^{s-1}$.
Since $[-\frac{1}{x}]$ is a concave expression for $x$ when $x>0$, its Taylor expansion is greater than or equal to the original value. Therefore, we have the following  inequalities:
\begin{equation}
	\begin{aligned}
		&G((B_k)^{s-1},(f_k)^{s-1},(D^{\mathbf{SC}^3}_k)^{s-1}|(D^{\mathbf{SC}^3}_k)^{s-1})\\
		=&G((B_k)^{s-1},(f_k)^{s-1},(D^{\mathbf{SC}^3}_k)^{s-1}) \\
		\geqslant&G((B_k)^{s-1},(f_k)^{s-1},(D^{\mathbf{SC}^3}_k)^{s-1}|(D^{\mathbf{SC}^3}_k)^{s-2})\\
		 \geqslant &0.
	\end{aligned}
\end{equation}
This shows that  $\{(B_k)^{s-1},(f_k)^{s-1},(D^{\mathbf{SC}^3}_k)^{s-1}\}$ is a feasible solution in the $s$-th iteration. As a result, $\sum\limits_{k=1}^{K} (l_k)^{s-1}$ is an achievable objective in the $s$-th iteration. Since in the $s$-th iteration, $\{(B_k)^s,(f_k)^s,(D^{\mathbf{SC}^3}_k)^s\}$ is the optimal solution, it satisfies that $\sum\limits_{k=1}^{K} (l_k)^s\leqslant\sum\limits_{k=1}^{K} (l_k)^{s-1}$. According to the monotone bounded theorem, the proposed iterative algorithm is assured to be  convergent.

\section{Proof of \bf{Theorem 2}}
(PB-3) is a convex optimization problem and satisfies the Slater condition, which guarantees strong duality. As a result, the optimal solution to the dual problem is the same as that of the primal problem. Therefore, we solve the dual problem of (PB-3) to obtain the closed-form solution for the bandwidth. The dual problem is formulated as follows:
	\begin{subequations}
		\begin{align}
			\mbox{(PB-4)} \  \max\limits_{\lambda} \min\limits_{\mathcal{B}} &\sum_{k=1}^K\frac{n N \!\left( \mathbf{v}_k\right)|\det \mathbf{M}_k|^\frac{1}{n}} {2^{\frac{2}{n}(T_kB_kr_k^{\text{comm}}-\log_2|\det \mathbf{A}_k|)}} \nonumber\\
			&\quad \quad +\lambda(\sum\limits_{k=1}^KB_k-B_{\max})\\
			\text{s.t.} \ \
			&\lambda\geqslant 0,  \label{B34}
		\end{align}
	\end{subequations}
	where $\lambda$ is the Lagrange multiplier. The Karush-Kuhn-Tucker (KKT) condition of (PB-4) is given by
	\begin{subequations}
	\begin{align}
			& \frac{\partial}{\partial B_k}\bigg[\sum\limits_{k=1}^K\frac{n N \!\left( \mathbf{v}_k\right)|\det \mathbf{M}_k|^\frac{1}{n}}  {2^{\frac{2}{n}(T_kB_kr_k^{\text{comm}}-\log_2|\det \mathbf{A}_k|)}} \nonumber\\
			&\quad \quad \quad \quad \quad +\lambda(\sum\limits_{k=1}^KB_k-B_{\max})\bigg]=0 \label{B41}\\
			&\sum\limits_{k=1}^KB_k-B_{\max}= 0 \label{B43}\\
			&\lambda(\sum\limits_{k=1}^KB_k-B_{\max})=0\\
			&\lambda\geqslant 0,
		\end{align}
	\end{subequations}
	where the constraint \eqref{B43} is an equality since the available bandwidth must be fully utilized at the optimal solution.
	By calculating \eqref{B41}, the optimal bandwidth is expressed as a function of $\lambda$,
	\begin{equation}
		\label{B42}
		(B_k)^*=\frac{n}{2T_kr_k^{\text{comm}}}\log_2(\frac{ 2\mbox{In}2N(\mathbf{v}_k)|\det\mathbf{A}_k|^{\frac{2}{n}}T_kr_k^{comm}}{\lambda}).
	\end{equation}
 We further substitute \eqref{B42} into \eqref{B43} to get the following expression of the Lagrange multiplier,
	\begin{equation}
		\label{B44}
		\begin{aligned}
			&\log_2(\lambda)=\\
			&\frac{\sum\limits_{i=1}^{K}\frac{n}{2T_ir_i^{\text{comm}}}\log_2( 2\text{In}2N(\mathbf{v}_i)|\det\mathbf{A}_i|^{\frac{2}{n}}T_ir_i^{\text{comm}})-B_{\max}}{\sum\limits_{i=1}^{K}\frac{n}{2r_i^{\text{comm}}T_i}}.
		\end{aligned}
	\end{equation}
By back-substituting the above expression into \eqref{B42}, we thus obtain the closed-form solution for the inter-loop bandwidth allocation, as given in \eqref{q1} in {\bf{Theorem 2}}.

\end{document}